\documentclass[letter paper, 10 pt, conference]{ieeeconf}

\IEEEoverridecommandlockouts                              % This command is only needed if 
\usepackage{textcomp}

\usepackage{epsfig} % for postscript graphics files
\usepackage{psfrag}
\usepackage{cite}
\usepackage{amsmath} % assumes amsmath package installed
\usepackage{amssymb}  % assumes amsmath package installed
\usepackage{dsfont}  % assumes amsmath package installed
\usepackage{todonotes}
\usepackage{graphicx}
\usepackage{subcaption}
\usepackage[noend]{algorithmic}
\usepackage{algorithm, ulem}
\usepackage{cleveref}

\usepackage{color}
\definecolor{red}{rgb}{1,0.2,0.2}
\definecolor{green}{rgb}{0.2,1,0.5}
\definecolor{blue}{rgb}{0,0,1}
\definecolor{lightblue}{rgb}{0.3,0.5,1}

\usepackage{amsthm}

\newtheorem{lemma}{Lemma}
\newtheorem{theorem}{Theorem}
\newtheorem{corollary}{Corollary}
\newtheorem{definition}{Definition}
\newtheorem{assumption}{Assumption}
\theoremstyle{remark}
\newtheorem{remark}{Remark}

\newcommand{\T}{^{\mbox{\tiny \sf T}}}

\newcommand{\R}{\mathbb{R}}

\newcommand{\N}{\mathcal{N}}
\newcommand{\bbN}{\mathbb{N}}
\newcommand{\xad}{\hat{x}^{\rm ad}}
\newcommand{\ead}{e}

\newcommand{\sad}{\Sigma}

\DeclareMathOperator{\E}{\mathbb{E}}
\DeclareMathOperator{\bbP}{\mathbb{P}}

\newcommand{\inprod}[2]{\left< #1, #2 \right>}
\title{\bf Ensuring System-Level Protection against Eavesdropping Adversaries in Distributed Dynamical Systems}
\author{Dipankar Maity and Van Sy Mai
	\thanks{D. Maity is an Assistant Professor in the Electrical and Computer Engineering department at the University of North Carolina at Charlotte,  NC, 28223, USA. Email:
		{\small dmaity@uncc.edu}}
	\thanks{V. S. Mai is with the National Institute of Standards and Technology, MD 20899. Email:
		{\small vansy.mai@nist.gov}}%
	}

\begin{document}

\maketitle

\begin{abstract} 
    In this work, we address the objective of protecting the states of a distributed dynamical system from eavesdropping adversaries. We prove that state-of-the-art distributed algorithms, which rely on communicating the agents' states, are vulnerable in that the final states can be perfectly estimated by any adversary including those with arbitrarily small eavesdropping success probability.
    While existing literature typically adds an extra layer of protection, such as encryption or differential privacy techniques, we demonstrate the emergence of a fundamental protection quotient in distributed systems when innovation signals are communicated instead of the agents' states.
\end{abstract}

\section{Introduction}

Privacy against eavesdropping adversaries has been a major concern in distributed systems \cite{li2021privacy}. To protect data from eavesdropping adversaries one needs to deploy some mechanism where the eavesdroppers are not able to perfectly decode the data from the intercepted communications. 
Often, these mechanisms are to be deployed without having information about the adversary's capability and knowledge. 
In this work, we consider eavesdropping adversaries and protection against such adversaries \cite{anand2005quantifying, zhang2010p} for a class of systems following the setup of \Cref{fig:setup}.

% (1) need for optimization: motivation/application\\
In general, there are mainly two techniques for dealing eavesdropping adversaries, namely differential privacy \cite{dwork2008differential, han2016differentially} and secure multiparty computation \cite{shoukry2016privacy, lu2018privacy}. 
Differential privacy is a noncryptographic method for preserving privacy by carefully adding noise to exchanged messages. It is commonly used due to its computational simplicity. However, there is an inherent trade-off between privacy and accuracy that one has to take into account due to the nature of the method that intentionally adds noise to communication. 
Secure multiparty computation, on the other hand, refers to cryptographic techniques to ensure privacy in a distributed network, where the goal is to 
evaluate a function of a number of parties’ private data without revealing each party’s data to others; see, e.g., \cite{lu2018privacy, shoukry2016privacy}. %\cite{lu2018privacy, shoukry2016privacy, wang2011secure}. 
As most of secure multiparty computation protocols trade algorithmic 
complexity for security, they may not be suited for many practical applications involving, e.g., systems with limited resources or subject to hard real-time constraints.

\begin{figure}
    \centering
    \includegraphics[width=1\linewidth]{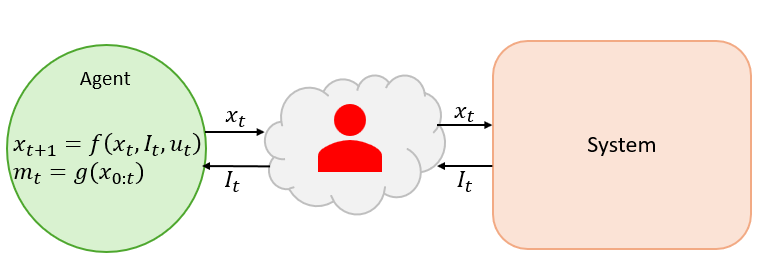}%Setup.png
    \caption{Problem setup}
    \vspace{-3mm}
    \label{fig:setup}
\end{figure}

Besides, existing approaches to address privacy concerns of distributed systems has mainly focused on the protection of the initial states of the agents \cite{mo2016privacy,charalambous2019privacy,manitara2013privacy,altafini2019dynamical}. 
However, in many distributed systems (e.g., distributed optimization \cite{Bertsekas89}, rendezvous problems \cite{dimarogonas2007rendezvous}, synchronization \cite{maggs2012consensus}, federated learning \cite{McMahan2017}) initial states often are of less importance and sometimes chosen arbitrarily; see e.g., \cite{shi2015extra, nedic2017achieving}. 
Instead, the final state of the agents are more important because they represent the solution of a certain decision making problem. 
For example, in \textit{networked consensus problems} \cite{olfati2004consensus}, an agent interacts with its neighbors and $I_t$ contains the states of the neighboring agents. Here, the objective is to agree on the final state $x^*$ to which all the agents' states will converge.  
\textit{Distributed consensus optimization} is a variant of the network consensus problems where a group of agents exchange their local information to collaboratively optimize a network-wide objective function, and 
In the context of \textit{federated learning},  where each client/agent shared its local weights of a trained neural network with a server for the purpose of aggregating the distributed information. The server sends the aggregated weight ($I_t$) back to the client for further update. In these cases, if the final state of one agent is eavesdropped, then so is that of the whole system. This requires modifications in privacy metrics as well as in algorithms/methods to achieve privacy.

\textit{Contributions:} The main contribution of this work is the analysis of system-level protection against an eavesdropping adversary under an \textit{innovation-sharing communication} protocol.
We derive an analytical expression for the achievable protection against a class of eavesdropping adversaries and demonstrate how the eavesdropper's capabilities affect this protection.
Our analysis reveals a fundamental connection between the achievable protection and the total quadratic variation of the agent's state trajectory.
By leveraging this analysis, we then demonstrate how the proposed method can be applied to protect the solutions of distributed optimization problems. 
To this end, we develop a Distributed Innovation-Sharing Consensus Optimization (DICO) algorithm. 
We also discuss the effects of the algorithm's parameters on protection and convergence speed, as well as their trade-offs.

\textit{Organization} We formally state the problem in \Cref{sec:formulation}, and the adversary's eavesdropping model is discussed in \Cref{Sec:adversarial_model}. 
The system-level protection against such adversaries is analyzed in \Cref{sec:mainAnalysis}, where we leverage the protocol of sharing the state increment (i.e., $m_t = x_t - x_{t-1}$) instead of the true state $x_t$ as an effective means of protection---such protocols are often categorized as \textit{innovation sharing schemes}. 
In \Cref{sec:distributedOptimization} we investigate distributed consensus optimization problems, as a special case of our developed theory and analysis. 
The evidence of a system-level privacy is demonstrated using numerical simulation on a distributed optimization problem in \Cref{secSimulation}.
The effects of certain hyperparameters of the optimization algorithm on the achieved protection is thoroughly discussed in that section. 
% Finally, we conclude the paper in \Cref{secConclusion}.

\textit{Notation:} 
% $\mathbb{R}$ and $\mathbb{R}_+$ denote the 
% sets of real numbers and nonnegative real 
% numbers, respectively. 
Let $\bar{a} = 1-a$ for any $a\in \mathbb{R}$. For a matrix $A=[a_{ij}]$, let $a_{ij}$ denote 
its $(i,j)$ element, and $A\T$ its transpose. 
A directed graph $\mathcal{G} \!=\! (\mathcal{V}, 
\mathcal{E})$ consists of a set of nodes $\mathcal{V}$ 
%$\mathcal{V} = \{1,2,...,N\}$ 
and a set $\mathcal{E}  \subseteq  \mathcal{V} \times  
\mathcal{V}$ of directed edges. 
A directed path is a  sequence of edges in the form $\big( (i_1, i_2), (i_2, i_3),..., (i_{k-1}, i_k) \big)$. 
The graph 
$\mathcal{G}$ is strongly connected if there is a 
directed path from each node to any other node. 
Node $j$ is an in-neighbor (respectively, out-neighbor) to node $i$ if $(i,j) \in \mathcal{E}$ (respectively, $(j,i) \in \mathcal{E}$).
For each node $i$, we use $\mathcal{N}_i$ and $\mathcal{N}^i$ to denote the sets of its in-neighbors and out-neighbors, respectively. Assume that $i \in \mathcal{N}_i$ and $i \in \mathcal{N}^i$. 

\section{Problem Formulation} \label{sec:formulation}

We consider a scenario where an agent interacts with a system over a compromised communication channel (see \Cref{fig:setup}). 
The agent is required to transmit its states to, and receive data from, the system at every time instance over this channel.
The communication channel is compromised due to the presence of an eavesdropper that can intercept the incoming and outgoing messages of this agent with probability $\gamma \in (0, 1)$.
The eavesdropper's objective is to estimate the agent's state $x_t$ as closely as possible. 
In some applications, the eavesdropper's objective is to only estimate the `final state' of the agent, i.e., $\lim_{t\to T} x_t$, where $T$ could be finite or infinite. 
The agent may not be aware of the presence of such adversaries. 

\subsection{Agent Dynamic Model and Assumption} 
The agent follows the dynamics
\begin{align} \label{eq:dyn}
    x_{t+1} &= f(x_t, I_t, u_t),\\
    m_t &=g(x_t, \ldots, x_0)
\end{align}
where $x_t \in \R^n$, $I_t \in \R^p$ and $u_t\in \R^m$ are the agent's states, received data, and control input, respectively, at time $t$. Here, $m_t$ is the message that is sent from the agent to the system at time $t$ and $g$ in general can be a function that depends on the agent's states up to time $t$.

\begin{assumption} \label{assm:state_convergence}
    The agent state following dynamic \eqref{eq:dyn} converges for any $x_0$. 
\end{assumption}

This assumption means that we consider only stable dynamics, which applies to numerous practical applications as mentioned in the previous section. This is also a key difference between our model and that in \cite{tsiamis2019state}, which is an unstable linear dynamics instead.
%=====================================================================
% \subsection{Adversary Model}
\subsection{Eavesdropping Adversaries} \label{Sec:adversarial_model}
Assume that the adversary eavesdrops all the outgoing (and incoming) transmissions of the agent.
Let $\hat{x}_t$ denote the adversary's estimate at time $t$. 
The eavesdropping mechanism is probabilistic and may lead to failed interception of the transmitted messages, similar to the models in\cite{li2018novel, tsiamis2019state}.
The success rate of eavesdropping depends on several factors including, e.g.,  
signal-to-interference-plus-noise ratio, 
% (SINR)  \cite[\S 4]{li2018novel},  
channel condition, 
% \cite[\S II-B]{mucchi2017new}, 
and  directionalities of transmitting and receiving antennas.  
% \cite[\S 2.3]{li2016analytical}.
This limitation in eavesdropping capability is often modeled by a randomness in the eavesdropping outcome.

Let $\mu_t$ be a Bernoulli random variable such that
\begin{align} \label{eq:mu_bernoulli}
    \mu_t=\begin{cases} 1,\!\! &\text{successful eavesdropping at time $t$},\\
    0, &\text{otherwise}.
    \end{cases}
\end{align}
  The random variables $\mu_t$ and $\mu_s$ are independent of each other for all $t\ne s$ and $\mu_t$ does not depend on the states of the physical system $\{x_s\}_{s\in \bbN_0}$.
  We denote $\bbP(\mu_t=1) \triangleq \gamma = 1-\bbP(\mu_t=0)$, for some $\gamma \in (0, 1)$. 
  We exclude $\gamma =1$ since, in this case, the adversary intercepts everything and hence, no   \textit{protection} is achievable. 
  Similarly, we exclude $\gamma = 0$ since it implies that no adversary is present. 
  
 \begin{assumption} \label{assm:message_knowledge}
     The adversaries know the form of the exchanged messages.  
   For example, if the physical system exchanges $m_t = g(x_t)$ at time $t$, then the adversary knows the function $g(\cdot)$ and intercepts $m_t$.
 \end{assumption}

Assumption~\ref{assm:message_knowledge} implies that the adversary knows whether true states are communicated or not. 
Here, $g(\cdot)$ could be a deterministic or randomized function (e.g., quantization, encryption, adding noise to $x$). 
In general, $g$ may also depend on the past $x_{t-1},\ldots,x_0$. 
In this paper we consider the form $g(x_t, x_{t-1}) = x_t - x_{t-1}$ 
and demonstrate the benefits of this simple form in retaining privacy.

In this paper, we embark on this direction by studying the following simple class of adversary dynamics.
Let $m_t$ denote the message that the physical system broadcasts at time $t$, and $z_t$ is the adversary's estimate of $m_t$. Consider the following class of estimation dynamics 
\begin{align}
    z_t = \bar\mu_t b_tz_{t-1} + \mu_tm_t, \quad \forall t\ge 0, \label{eq_adversary_state_general}
\end{align}
for some initial state $z_{-1}$ chosen by the adversary (which can be a random vector) and a sequence of weights $\{b_t\}_{t\ge 0}$ designed by the adversary. Here, for any $t\ge 0$, 
if $b_t=0$, then 
    $z_t = \mu_tm_t$,  
    i.e., the adversary simply takes the eavesdropping outcome at time $t$ as its update, without considering past information.
On the other hand, if $b_t\neq 0$, then  $z_t = m_t$ when $\mu_t = 1$ and $z_t 
 = b_t z_{t-1}$ otherwise. In other words, if the current interception is unsuccessful, the adversary uses its last estimate $z_t$ with some weight $b_t$. 
    This is also equivalent to using the last successfully intercepted message, say $m_{\tau}$, with the weight $ \prod_{k=\tau+1}^t b_k$, considering the number of recent unsuccessful attempts.

In general, the time-varying weight $b_t$ depends on the adversary's knowledge and possibly on $\{\mu_0, \dots, \mu_t\}$. 
This is an open problem and left for future work.
To continue the analysis, we instead focus on the following special case \begin{align}
    b_0 = 1, \quad b_t = b \in \mathbb{R}, \quad \forall t\ge 1.\label{eq_fixed_weight_b}
\end{align}
Here, we assume $b_0 = 1$ simply to decouple $b$ and $z_{-1}$.

Note that the state $z$ will carry different meanings depending on the type of exchanged messages $m$. The adversary's goal is to use $z$ to estimate the convergence point of $x$.

\subsection{$\epsilon$-Protection Against Adversaries}
Let $\hat{x}_t$ denote the adversary's estimate of the state at time $t$ and let $e_t\triangleq \hat{x}_t - x_t$ denote the corresponding estimation error. 
Since the eavesdropping success is random, the estimated state $\hat{x}_t$ and  the estimation error $e_t$ are random variables for all $t$.
We consider the following distortion based metric to quantify protection against the adversary.
\begin{definition}\label{Df:e-protected}
The agent's state is $\epsilon$-\textit{protected} against the eavesdropping adversary for some $\epsilon > 0$ if  
\begin{align} \label{eq:df_ep}
     \liminf_{t\to \infty}   \E \big[ \| e_t\|^2 \big] \ge \epsilon.
\end{align} 
\end{definition}
The proposed protection metric is similar in principle to the distortion metric proposed in \cite{agarwal2020distortion}. 
Here, we use the second moment to quantify the quality of protection since it is tightly coupled to the entropy-power of the random variable $e_t$. 
In particular, a lower bound on $\E\big[ \| e_t\|^2 \big]$ immediately provides a lower bound on the entropy-power of $e_t$ and consequently a lower bound on the randomness of $e_t$.\footnote{ 
 Entropy power of a random vector $X \in \R^m$ is $N(X) \triangleq \frac{1}{2\pi e}e^{\frac{2}{m}h(X)}$, where $h(X)$ is the entropy of $X$. For any random vector $X$, we have $N(x) \ge {\rm Var}(X)$ where the equality holds only for Gaussian random vectors with uncorrelated components. 
 }

\subsection{Innovation-shared Communication}

As we will see, there will be no protection when the agent exchanges its true state, i.e., $m_t= x_t$. Thus, we propose to use an \textit{innovation-shared communication} protocol where the agent shares $\xi_{t-1}$ instead of $x_t$ at time $t$. 
In this protocol, $\xi_{t-1}$ is defined as follows:
\begin{align} \label{eq:innovation_signal}
    \xi_{t-1} =  x_t - x_{t-1}, \qquad \qquad \xi_{-1} = x_0.
\end{align}
To the best our knowledge, this particular communication protocol was first used in \cite{tsiamis2019state} to achieve protection against eavesdropping adversaries for an remote estimation application. 
Later, this protocol was used in \cite{arXivMaity2021} to study the protection against eavesdropping adversaries in networked consensus problems. 
Although the use of \eqref{eq:innovation_signal} is relatively new in understanding privacy, however, it has been used in other control problems such as in quantized optimal control. 
The benefit of using this communication protocol is that the environment can perfectly decode $x_t$ from $\xi$'s by using the relationship $x_t = \sum_{k = -1}^{t-1} \xi_k$.

\begin{remark}
Although our communication model \eqref{eq:innovation_signal} appears simple, it has several benefits. First and foremost, as we will show later, it is sufficient for rendering setup of \Cref{fig:setup} unprotected when the agent shares $x_t$, thus motivating our modifications to those algorithms to enhance their protection. 
Second, it enables us to conduct a rigorous analysis and provide insights into the effect of adversary's model parameters to protection, serving as basis for further extensions.
\end{remark}

\textbf{Problem Statement:} The objective of this paper is to analyze and derive the system level protection of system \eqref{eq:dyn} under the innovation-sharing scheme \eqref{eq:innovation_signal} against the eavesdropping adversaries in \Cref{Sec:adversarial_model}.

\section{System-Level Protection Analysis} \label{sec:mainAnalysis}
\subsection{State-based Communication}
In this section, we show that \textit{any} system in the form of \Cref{fig:setup} is unprotected against eavesdropping adversaries regardless of the eavesdropping probability $\gamma$. 
\begin{theorem} \label{thm:zero-protected}
The state of the agent is $0$-protected if $m_t = x_t$.
\end{theorem}

\begin{proof}
Let $b=1$ and $z_{-1}=0$. Then, \eqref{eq_adversary_state_general} simply becomes
$$
    z_{t+1} = \bar\mu_{t+1} z_t + \mu_{t+1} x_{t+1}
$$
where $\mu_t$ is as defined in \eqref{eq:mu_bernoulli}. In this case, $z_i$ is the adversary's estimate of $x_i$, and thus $\hat{x}_i^{\rm ad} = z_t.$ 
The estimation error $\ead_t\triangleq \hat{x}^{\rm ad}_t - x_t$ follows the dynamics $\ead_{t+1}=\bar\mu_{t+1} \big( \ead_t - \xi_t\big),$
where $\xi_t = x_{t+1}-x_t$.
The estimation error $\ead_i$ is a random process due to the presence of the Bernoulli random variables $\mu_t$. 
Let 
$\sad_t \triangleq \E[\ead_t\ead_t\T ]$. We have
\begin{align}
    \E[\ead_{t+1}] &=\bar\gamma \big(\E[\ead_t] -\xi_t\big) \label{E:ead_state_communication}\\
    \sad_{t+1}
            &= \bar\gamma \big(\sad_t - \E[\ead_t]\xi_t\T-\xi_t\E[\ead_t]\T + \xi_t\xi_t\T\big), \label{E:sad_state_communication} 
\end{align}
where we have used the fact that the random variables $\mu_t$ is  independent of $\{\mu_s\}_{s< t}$ and $\{x_s\}_{s\ge 0}$.
Note that $\lim_{t\to \infty}\xi_t = 0$ since $x_t$ converges (Assumption~\ref{assm:state_convergence}). 
Consequently, from \eqref{E:ead_state_communication}--\eqref{E:sad_state_communication}, we obtain $\E[\ead_t]\to 0$ and $\sad_t\to 0$ as $t\to \infty$. 
This completes the proof.
 \end{proof} \vspace{4 pt}
 
The intuition behind this result is that, an adversary is able to intercept a transmission far in the future with probability $1$, and hence, it obtains $x_t$ for a large enough~$t$.\footnote{
 Let us define an event $\mathcal{A}_T =\{\exists~t\ge T  \text{  such that  } \mu_t =1\}$.
 The complementary event $\mathcal{A}^c_T = \{\forall ~t\ge T, ~~\mu_t = 0\}=\cap_{t\ge T} \{\mu_t =0\}$.
 Therefore, $\bbP(\mathcal{A}_T^c)= \prod_{t\ge T}\bbP(\mu_t = 0)=\prod_{t\ge T}(1-\gamma_i) = 0$ since $1 > \gamma_i > 0$.
 Consequently, $\bbP(\mathcal{A}_T) = 1$ for all $T < \infty$. 
 The event $\mathcal{A}_T$ denotes a successful interception at time later than $T$.
 }
In fact, if $m_t = g(x_t)$ and $x_t$ can be inferred/decoded from $m_t$, then the agent state is also not protected.
%TODO: add more explanation here. 

\subsection{Innovation-based communications}
Let the agent exchange state increments $\xi$ instead of actual state $x$. Thus,  
\begin{align}
    x_{t+1} = x_t + \xi_t, \quad \forall t\ge -1, \label{eq_x_t1}
\end{align}
where we define $x_{-1} = 0$. 
In this case the exchanged message at time $t$ is $m_t = \xi_{t-1}$.
Thus, \eqref{eq_adversary_state_general} becomes
\begin{align}\label{eq_z_t1}
    z_{t} = \bar\mu_{t} b  z_{t-1} +  \mu_{t}\xi_{t-1}.
\end{align}
Here, $z$ is the adversary's estimate of $\xi$ and thus
\begin{align} 
\xad_{t+1} = \xad_t + z_{t+1}, \quad \forall t\ge -1, \label{eq_xhat_t1}
\end{align}
with $\xad_{-1} = 0$. 
Then the error $e_t = \xad_t - x_t$ satisfies
\begin{align}
    e_{t+1} &= e_t + z_{t+1} - \xi_t \label{eq_et1}
\end{align}
with $e_0 =z_0-\xi_{-1} = \bar\mu_0(z_{-1}-x_0)$.

Next, taking expectation on both sides of \eqref{eq_z_t1} yields
\begin{align}
    \E z_{t+1} 
    &= c \E z_t + \gamma\xi_t = c^{t+1} \E z_0 + \gamma \omega_t, \label{eq_TE_z}
\end{align}
where $c:= b\bar{\gamma}$, $\omega_t = \textstyle \sum_{k=0}^t \xi_k c^{t-k}$, and $\E z_0 = \gamma x_0 + \bar{\gamma}\E z_{-1}$ with $z_{-1} = z_{-1}$.  
Furthermore, 
\begin{align} \label{eq_TE_znorm}
    \E \|z_{t+1}\|^2 & = bc \E\|z_t\|^2 +  \gamma\|\xi_t\|^2 
\end{align} 
with $\E \|z_0\|^2 \!=\! \bar{\gamma}\E\|z_{-1}\|^2 \!+\! \gamma\|x_0\|^2$. 
Thus, taking expectations on both sides of \eqref{eq_et1} yields
\begin{align}
    \E e_{t+1} 
    & \overset{\eqref{eq_TE_z}}{=} \E e_t + c \E z_t - (1-\gamma)\xi_t \label{eq_TE_et}
\end{align}
and $\E e_0 = \bar{\gamma}(\E z_{-1}-x_0)$.
Using the above relations, 
we can find the limit $\E \|e_{\infty}\|^2 := \lim_{t\to \infty} \E\|e_t\|^2$ as follows.

\begin{theorem}\label{thm_protection_infinite_horizon}
Suppose $Q\!=\!\sum_{t}\|\xi_t \|^2 \!<\! \infty$ and $\bar{\gamma}b^2 \!<\! 1$. 
Let $\rho \!=\! \frac{b\gamma}{\bar{c}}$, $\nu \!=\! b \!-\!\bar{\gamma} \!-\! \rho\gamma$, and $x^*=\lim_{t\to\infty}x_t$. Then,
\begin{align} 
    \textstyle \frac{\bar{c}}{\bar{\gamma}}\E \|e_{\infty}\|^2 
    &= \textstyle \frac{b(b+1)}{1-bc}(\bar{\gamma}\E\|z_{-1}\|^2 + \gamma\|x_0\|^2 ) - b\textstyle \|x_0\|^2 \nonumber\\
    &\quad + \E\|z_{-1}-x_0\|^2 
    +(\textstyle  \frac{b(b+1)}{1-bc} + \bar{\rho})\gamma Q -2\rho R \nonumber\\
    &\quad- \textstyle \frac{2\bar{b}\bar{\gamma}}{\bar{c}}\inprod{\E z_{-1}-\bar{b}x_0}{d_0}-\nu\|d_0\|^2, \label{eq_themorem_main}
\end{align} 
where $d_0 \!=\! x^*\!-\!x_0$ and $R \!=\! \sum_{t\ge 0} \E z_t^T\xi_t$. 
\end{theorem}
\begin{proof}
See Appendix~\ref{proof_main_theorem_infinite_horizon}.
\end{proof}

\begin{remark}
    This result holds without any assumption on the agent's dynamics $f$, its control objective, or the structure of $I_t$, which makes out analysis applicable to a wide range of problems. 
    %The agent's dynamics and control 
    Here, $f$ and $u_t$
    indirectly affect the protection amount through the variable $\xi_t$. 
    Different choices for $f, u_t$ or $I_t$ will affect the norm $\|\xi_t\|$ differently, and consequently, resulting in different amounts of protection. 
\end{remark}

Let us note the following. First, conditions $Q<\infty$ and $\bar{\gamma}b^2<1$ are sufficient for the stability and boundedness of the systems in \eqref{eq_TE_z}--\eqref{eq_TE_et}, and hence the finiteness of $\E \|e_{\infty}\|$ given above.\footnote{Note that $\bar{\gamma}b^2<1$ also implies $|c|<1$. Although $\xi_t\to 0$ and $\sum_t \xi_t = x^*-x_0$, which is finite, it does not imply that $\sum_t \|\xi_t\|^2<\infty$; to see this, consider, e.g., $\xi_t = (-1)^t/\sqrt{t}$ for $t\ge 1$.}
Second, since $\E\|z_{-1} - q\|^2 = \|\E z_{-1} - q\|^2 + \E\|z_{-1} - \E z_{-1}\|^2$ for any $q$, it follows that, to minimize $\E\|e_\infty\|^2$ the adversary must choose $\E\|z_{-1} - \E z_{-1}\|^2=0$, or equivalently, select $z_{-1}$ to be a deterministic quantity.
Third, $R$ depends linearly on $\E z_{-1}$ as follows
\begin{align*} 
&\textstyle \sum_{t}\!\inprod{\E\! z_t}{\xi_t} \overset{\eqref{eq_TE_z}}{=} \textstyle\sum_t\!\inprod{c^t\E \!z_0}{ \xi_t}  + \gamma\sum_{t}\! \inprod{\omega({t\!-\!1})}{\xi_t}\nonumber\\
& = \textstyle \inprod{\gamma x_0 + \bar{\gamma}\E z_{-1}}{\sum_t c^t\xi_t}  + \frac{\gamma}{c}\sum_{t, k<t} c^{t-k}\inprod{\xi_k}{\xi_t}. %\label{eq_R_exact}
\end{align*} 
However, this is rather complicated for computing $R$; in practice, we use the recursive form in \eqref{eq_TE_z} instead. Below, we provide a lower bound for the protection that does not depend on $R$ explicitly; see Appendix \ref{proof_coro_protection_LB} for a proof.

\begin{corollary}\label{coro_protection_LB}
For any $\eta>0$, let $h=\frac{b^2+b-|\rho|\eta^{-1}}{1-bc}$. Then
\begin{align}
    \textstyle \frac{\bar{c}}{\bar{\gamma}}\E \|e_{\infty}\|^2 
    &\ge \|z_{-1}-x_0\|^2  + h\bar{\gamma}\|z_{-1}\|^2+(h\gamma-b)\|x_0\|^2 \nonumber\\
    &\quad\textstyle +(\gamma\bar{\rho} + h\gamma - |\rho|\eta) Q  \nonumber\\
    &\quad\textstyle - \frac{2\bar{b}\bar{\gamma}}{\bar{c}}\inprod{ z_{-1}-\bar bx_0}{d_0} -\nu\| d_0 \|^2, \label{eq_protection_LB}
\end{align} 
where equality holds at $b=0$.
\end{corollary} 
\begin{proof}
    See Appendix \ref{proof_coro_protection_LB}.
\end{proof}

Since the lower bound given in \eqref{eq_protection_LB} is valid for any $\eta>0$, one may maximize the RHS in \eqref{eq_protection_LB} with respect to $\eta$ to obtain a tighter bound.
On the other hand, the adversary should minimize this lower bound by selecting $z_{-1}$ and $b$ appropriately. 
Since $z_{-1}$ affects the part $\|z_{-1}-x_0\|^2  + h\bar{\gamma}\|z_{-1}\|^2 - \frac{2\bar{b}\bar{\gamma}}{\bar{c}}\inprod{ z_{-1}}{d_0}$, an optimal $z_{-1}$ clearly depends on  $x_0$ and $x^*$. 
In absence of knowledge on $x_0$ and $x^*$, a rational choice is to pick $z_{-1} = 0$, which will minimize the worst-case value of $\|z_{-1}-x_0\|^2  + h\bar{\gamma}\|z_{-1}\|^2 - \frac{2\bar{b}\bar{\gamma}}{\bar{c}}\inprod{ z_{-1}}{d_0}$. 
Finding the optimal $b$ is even more complicated as it depends not only on $x_0$ and $x^*$ but also on $Q$, where $Q$ is affected by the dynamics of both the agent and the system. 
Finding the optimal $b$ appears to be equally challenging as finding $x^*$ directly. 
 In the following, we investigate two special cases where $b=0$ or $1$.

\vspace{3pt}
\subsubsection{The case $b=0$}
By \eqref{eq_themorem_main}, we have
\begin{align}\label{protection_b=0}
    \E\! \|e_{\infty}\|^2 
    = \gamma\bar{\gamma}^2\| d_0 \|^2 + \bar{\gamma}\|z_{-1} \!-\! x_0 \!-\! \bar{\gamma}d_0\|^2  \!+\! \gamma \bar{\gamma} Q.
\end{align}
Clearly, the first two terms depend on $x_0$ and $x^*$ but not $\{x_t\}_{t\ge 1}$ directly. 
On the other hand, for a fixed $\gamma\in (0,1)$, the last term $\gamma \bar{\gamma} Q$ depends not only on initial condition $x_0$ but also on the state's total quadratic variation $\sum_{k=0}^{\infty}\|x_{k+1}-x_k\|^2$. 
% , which is algorithm-dependent. 
This in turn has the following two consequences: 
(i) For given dynamics of the agent and the system, 
choosing $x_0$ far from the convergence point $x^*$ will yield better protection, and (ii) for a given $x_0$, a dynamic 
that produces a path with higher quadratic variation also has better protection. 
In the latter case, it is tempting to conclude that faster convergence yields a smaller protection level. 
However, it could happen that, starting from the same initial condition, a dynamic with faster convergence may exhibit more (and possibly larger) transient oscillations and thus incurs a higher quadratic variation, hence improved protection. 

Finally, we can quantify the amount of randomness in the adversary's estimates of $x$ for the case $b=0, z_{-1} =0$.
% \begin{remark}[$b=0, z_{-1} =0$ case] 
 Using \eqref{eq_TE_et}, one may write $\E e_{t+1} = \E e_t - \bar\gamma \xi_t = - \bar\gamma x_t$. 
 Furthermore, we also obtain that ${\rm{Var}} (e_t) = \E\|e_t\|^2 - \| \E e_t\|^2 = \gamma\bar\gamma \!\! \sum_{k=-1}^{t-1}\! \|\xi_k\|^2$.
  Consequently, the entropy power of $e_t$ is lower bounded by $\gamma\bar\gamma \!\! \sum_{k=-1}^{t-1}\! \|\xi_k\|^2$.
 While the first and second order moments partially characterize a random variable, the entropy power is a direct indication of its randomness. 
 Notice that $\lim_{t\to \infty} {\rm Var}(e_t) = \gamma\bar\gamma \sum_t \|\xi_t\|^2$ and therefore, the innovation-share communication scheme ensures a lower bound of $\gamma\bar\gamma \sum_t \|\xi_t\|^2$ on the asymptotic entropy power of the adversary error estimate.

\vspace{3pt}
\subsubsection{The case $b=1$} 

It is easy to see that  \eqref{eq_themorem_main} yields
\begin{align} \label{eq:b=1_case}
    \textstyle\frac{\gamma}{\bar \gamma}\E \|e_{\infty}\|^2 = (\frac{2}{\gamma}-1) \|z_{-1}\|^2 + 2\|x_0\|^2 + 2Q - 2R.
\end{align}
Clearly, $z_{-1} = 0$ is the optimal choice for the adversary. Moreover, the adversary in fact obtains an unbiased estimate of $x^*$ in this case. 
The last expression of $\E \|e_\infty\|^2$ can be further simplified and given in the following corollary; see Appendix \ref{proof_coro_protection_unbiased} for a proof.

\begin{corollary}\label{coro_protection_unbiased}
    If $b=1$ and $z_{-1} = 0$, then $\E e_\infty = 0$ and $\E \|e_{\infty}\|^2 =   \frac{\bar\gamma}{\gamma}(\|x_0\|^2 + \sum_{t\ge 0} \|\E z_t - \xi_t\|^2)$.
\end{corollary}

This result shows that, to obtain an unbiased estimate, the adversary must always use the last successfully intercepted message. Additionally, similar to the previous case, both $x_0$ and the mismatch between $\xi_t$ and $\E z_t$ affect the protection of the algorithm. 
% Clearly, when $\gamma = 1$, we have $\E \|e_\infty\|^2 = 0$.
Using the dynamic \eqref{eq_TE_z}, one may write 
$\E z_t - \xi_t = - \bar{\gamma}^{t+1} x_0 - \sum_{k=0}^{t} \bar\gamma^{t-k} (\xi_k - \xi_{k-1})$ and consequently, 
$\E \|e_\infty\|^2 = \frac{\bar\gamma}{\gamma}(\|x_0\|^2 + \sum_{t\ge 0} \| \bar{\gamma}^{t+1} x_0  + \sum_{k=0}^{t} \bar\gamma^{t-k} (\xi_k - \xi_{k-1}) \|^2)$.
This shows that the achieved protection is related to the variation in the innovation signal $\xi$ whereas in the $b=0$ case, it is related to the quadratic variation of the state $x$.
This is not surprising since for the case $b=1$, the adversary's estimate of $\xi_t$ depends on the last intercepted message $\xi_{t-k}$ for some $k\le t+1$. 
Therefore, the variation in $\xi$'s trajectory should directly affect the protection. 

An unbiased estimate may not be always preferable if the adversary's objective is to minimize the amount of protection. 
Based on the expressions of protection for both $b=0$ and $b=1$ along with $z_{-1}=0$, we notice that neither of them always dominates the other.
The best choice for $b$ depends on several parameters including the dynamics \eqref{eq:dyn}, the objective of the agent, and the system itself.  
Without such knowledge, the adversary is unable to determine which value of $b$ is the best to use. 
However, some qualitative analysis could be performed here; e.g., when $\gamma \approx 0$, it appears  beneficial to use $b=0$ than $b=1$ since in the latter case $\E \|e_\infty\|^2 \to \infty$ as $\gamma \to 0$. 
A thorough investigation on the choice of $b$ is a potential future direction to pursue.

\section{Application to Distributed Optimization} \label{sec:distributedOptimization}
Next we discuss application of our approach to distributed consensus optimization. 
Consider a network of $n$ nodes, where the underlying communication is characterized by a fixed directed graph $\mathcal{G} \!=\! (\mathcal{V}, 
\mathcal{E})$. 
The objective of all the nodes is to solve the following  problem in a distributed fashion:
	\begin{align} 
	\mathrm{minimize}_{x \in \mathbb{R}^m} \quad F(x) \triangleq \textstyle \sum_{i=1}^n f_i(x), \label{eqProblem}
	\end{align} 
where $f_i$ is the local cost function of node $i\in \mathcal{V}$.

We consider a class of first-order distributed algorithms for solving this problem shown in Algorithm~\ref{algDCO} below, 
where each node repeatedly updates its local state 
based on its local gradient and information exchanged  with its direct neighbors. The goal here is for all nodes to reach a consensus that is also an optimal solution of \eqref{eqProblem}.  
\begin{algorithm}
\caption{Distributed Consensus Optimization (DCO)} \label{algDCO}
\begin{algorithmic}[1]
\STATE Initialize $x_{i,0}$, $y_{i,0}$, step size $\alpha$ for all $i\in \mathcal{V}$\\ 
\FORALL{$t\ge 0$ and $i \in \mathcal{V}$}
\STATE send $x_{i,t}$ to neighbors $j\in \N^i$ 
\STATE $x_{i,t+1} = \textstyle\sum_{j \in \mathcal{N}_i} w_{ij} x_{j,t} - \alpha y_{i,t}$
\STATE $y_{i,t+1} = h_i(y_{i,t},\nabla f_i, x_{\mathcal{N}_i,t})$
\ENDFOR
% \ENDFOR
% \vspace{5pt}
\end{algorithmic}
\end{algorithm}

Here, $x_i \in \R^m$ is the local estimate of node $i$, $\alpha >0$ is some fixed step size,
 $w_{ij} \in [0,1]$ the weight associated with link $(i,j)$, and  $y_i$ a local estimate of global gradient $\nabla F$, which is updated using only available local information 
according to some mapping $h_i$. 
To implement the algorithm, it is important to note that, at every time step $t\ge 0$, each node $i\in \mathcal{V}$ needs to send its local estimate $x_{i,t}$ to its out-neighbors $\mathcal{N}^i$ and receive $x_{j,t}$ from its in-neighbors $j \in \mathcal{N}_i$.  
Here, we use $x_{\N_i}$ to denote the vector $[x_j: j \in \mathcal{N}_i] \in \R^{|\mathcal{N}_i|m}$ and $x$ to denote $[x_j: j \in \mathcal{V}] \in \R^{nm}$.

Next, we mention a few algorithms that belong to DCO. First, a version of the well-studied distributed (sub-)gradient method (see, e.g.,  \cite{nedic2009distributed}) can be obtained with $y_{i,t+1} = \beta_t \nabla f_i(x_{i,t+1})$ for arbitrary $x_{i,0} \in \mathbb{R}^m$  for all $i\in \mathcal{V}$,  and some diminishing step size sequence $\{\beta_t\}_{t\ge 0}$. 
% Here, $\nabla f_i$ can also be a subgradient when dealing with  nondifferentiable objective functions. 
The distributed dual averaging in \cite{duchi2011dual} also takes a similar form with $h_i$ involving a type of projection with respect to some proximal function. 
Second, the following choice 
        \begin{align} 
        \begin{split} \label{eq:extra}
             y_{i,t+1} = y_{i,t} &+ \nabla f_i(x_{i,t+1}) - \nabla f_i(x_{i,t}) \\
                &- \textstyle\frac{1}{2\alpha} \sum_{j \in \mathcal{N}_i} w_{ij} (x_{j,t} - x_{i,t})
        \end{split}
        \end{align}
corresponds to a variant of the algorithm in  \cite{shi2015extra}.

In this set up, each node is interacting with the system through its in- and out-neighbors. 
Here we consider a single adversary that can pick any node to eavesdrop.
To compute a most conservative estimate of the protection we consider the minimum of the protection of the nodes.
That is, we define 
\begin{align*}
    \textstyle \min_{i \in \mathcal{V}} \|e_{i,\infty}\|^2
\end{align*}
to be the protection of the network in this case. Note that the presence of multiple adversaries is a potential future research direction, especially when the adversaries can communicate with each other. 
Theorem~\ref{thm:zero-protected} immediately shows that Algorithm~\ref{algDCO} is not protected. 
To achieve system level protection with innovation-shared communication scheme, we propose the a modification in the next section.

\subsection{Optimization with Innovation Communication}
Our proposed modification to DCO is a new communication protocol, where nodes communicate \textit{innovation} values $\xi_{i,t}$ instead of their state values. 
Each agent needs the true state values of their in-neighbors to update their own states (see line 5 of \Cref{algDCO}). 
In absence of these true state values, the agents need to perform an extra step to locally compute their in-neighbors' states:
\begin{align*}
    \hat{x}^i_{\mathcal{N}_i,t} = \hat{x}^i_{\mathcal{N}_i,t-1} + \xi_{\mathcal{N}_i,t-1}
\end{align*}
where $\xi_{\mathcal{N}_i,t-1}$ is the received innovation signal and $\hat{x}^i_{\mathcal{N}_i,t}$ is the estimate of the in-neighbors' state  at time $t$.
The modified algorithm, named Distributed Innovation-shared Consensus Optimization (DICO), is presented in Algorithm~\ref{AL:ICC}.
\begin{algorithm}
\caption{\small DICO: Distributed Innovation-shared Consensus Optimization} \label{AL:ICC}
\begin{algorithmic}[1]
\STATE Initialize $\hat{x}^i_{\mathcal{N}_i,-1} \!=\! 0, \xi_{i,-1} \!=\! x_{i,0}$, $y_{i,0}$ appropriately\\
\FORALL{$t\ge 0$ and $i \in \mathcal{V}$}
% \FORALL{$i \in \mathcal{V}$}%{$i=1,2,\ldots, N$}
\STATE send ${\xi_{i.t-1}}$ to neighbors $j\in \N^i$

\STATE $ 	\hat{x}^i_{\mathcal{N}_i,t} = \hat{x}^i_{\mathcal{N}_i,t-1} + \xi_{\mathcal{N}_i,t-1}
$

\STATE	$
	\xi_{i,t} = \sum_{j \in \mathcal{N}_i} w_{ij} (\hat{x}^i_{j,t} - x_{i,t}) - \alpha y_{i,t}
	$
	
\STATE	$	x_{i,t+1} = x_{i,t} + \xi_{i,t} 	$

\STATE	$
	y_{i,t+1} = h_i(y_{i,t},\nabla f_i, \hat{x}^i_{\mathcal{N}_j,t}) 
	$
% \ENDFOR
\ENDFOR
\end{algorithmic}
\end{algorithm}

Compared to DCO in terms of memory requirement, DICO further requires each node $i$ to maintain an estimate $\hat{x}^i_{\mathcal{N}_i}$ of its neighbors' states. % for all $i \in \V$. 
However, it is important to note that DICO has the same communication overheads per iteration as DCO. 
More importantly, \emph{DCO and DICO are indeed equivalent}, and thus the convergence property of DCO carries over to DICO. 
\begin{theorem}
The convergence of Algorithm~\ref{AL:ICC} is the same as Algorithm~\ref{algDCO} if the same step size $\alpha$ is used.
\end{theorem}
\begin{proof}
First, note that the local estimates are in fact exact at any time $t$, i.e., $\hat{x}^i_{\mathcal{N}_i,t} = x_{\mathcal{N}_i,t}$. This can be seen by comparing their dynamics, where $\xi_{i,-1} = x_{i,0}$ and 
$\hat{x}^i_{\mathcal{N}_i,t} = \textstyle \sum_{k=0}^{t-1} \xi_{\mathcal{N}_i,k-1}  = x_{\mathcal{N}_i,t}.$ 
Thus, DICO is equivalent to
\begin{align*}
	x_{i,t+1} &= x_{i,t} + \textstyle \sum_{j \in \mathcal{N}_i} w_{ij} (x_{j,t} - x_{i,t}) - \alpha y_{i,t}\\
	y_{i,t+1} &= h_i(y_{i,t},\nabla f_i, x_{\mathcal{N}_i,t})
	\end{align*}
which are identical to DCO. %\eqref{eqDGD}--\eqref{eqGradEst}. 
The proof is completed.
\end{proof}

Clearly, the innovation-shared method does not alter the convergence of DCO, which achieves exact solutions and  is in contrast to existing methods based on differential privacy, %distributed optimizations 
where the accuracy of the solution is negatively affected by the amount of privacy. 
%(see, e.g., \cite{huang2015differentially}). 
Unlike DCO, which is 0-protected against a single adversary, DICO provides a certain level of protection for each node as analyzed section~\ref{sec:mainAnalysis}. 

\section{Simulation Results}\label{secSimulation}
We consider a logistic regression problem  as follows
\begin{align}
    \min_{x\in \R^m} \textstyle \sum_{i=1}^n \Big[  \frac{\sigma}{2n}\|x\|^2 +\sum_{j=1}^{D_i}\ln (1+e^{-(a_{ij}\T x)\ell_{ij}})\Big], \label{eq:regression}
\end{align}
where each node $i$ has access to $D_i$ samples of training data $(a_{ij}, \ell_{ij}) \in \R^m\times \{-1,1\}$ for $j=1,\ldots,D_i$ and $\sigma>0$ is a regularization parameter. 
Here, $a_{ij} \in \R^m$ includes $m$ features of the $j$-th sample of node $i$, and $\ell_{ij} $ is the corresponding label. 
Clearly, \eqref{eq:regression} is in the form of \eqref{eqProblem} with $f_i(x) =  \frac{\sigma}{2n}\|x\|^2 +\sum_{j=1}^{D_i}\ln\left(1+\exp{(-(a_{ij}\T x)\ell_{ij})} \right)$ for all $i$. We consider $n=10$, $m=3$ and $D_i=10$ for all $i$. We generated the graph randomly while ensuring that it is connected and the weight matrix $[w_{ij}]$ is doubly stochastic. 
The dynamic of $y_i$ in DICO follows~\eqref{eq:extra}.

\begin{figure}
    \centering
    \begin{subfigure}{0.48 \linewidth}
         \centering
         \includegraphics[trim = 200 320 210 330, clip, width = \linewidth]{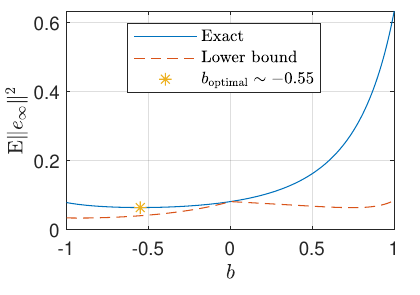}
         \caption{$\sigma = 1$} \label{fig:sigma=1}
    \end{subfigure}
    \begin{subfigure}{0.48 \linewidth}
         \centering
         \includegraphics[trim = 200 320 210 330, clip, width = \linewidth]{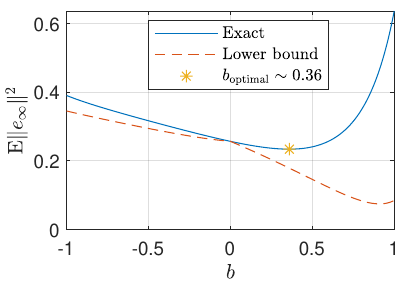}
         \caption{$\sigma = 0.01$}  \label{fig:sigma=.1}
    \end{subfigure}
     \caption{Exact protection and lower bound in \eqref{eq_protection_LB} with $\eta=1$.}
    \label{fig:varyb}
\end{figure}
\begin{figure}
    \centering
    \begin{subfigure}{0.48 \linewidth}
         \centering
         \includegraphics[trim = 200 320 210 330, clip, width = \linewidth]{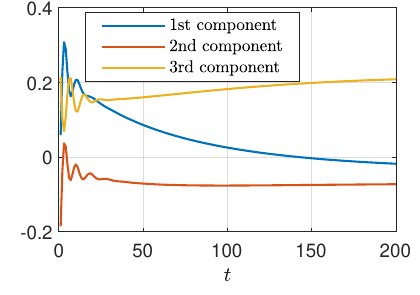}
         \caption{$\sigma = 1$} \label{fig:state_sigma=1}
    \end{subfigure}
    \begin{subfigure}{0.48 \linewidth}
         \centering
         \includegraphics[trim = 200 320 210 330, clip, width = \linewidth]{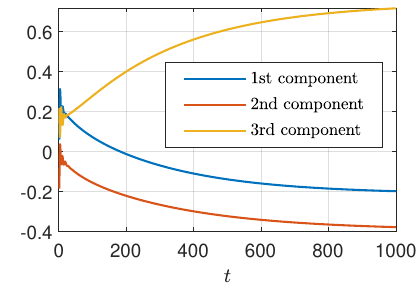}
         \caption{$\sigma = 0.01$}  \label{fig:state_sigma=.1}
    \end{subfigure}
    \caption{The state trajectory of the first agent.}
    \label{fig:varyb-state-trajectory}
    \vspace{-3mm}
    % \vspace{-5 mm}
\end{figure}

In the first experiment, we investigate how $b$ affects the protection. 
To that end, we randomly generated the initial states of the nodes (hereafter denoted as $x_{\text{init}}$) and considered $\alpha = 0.01$ and the eavesdropping probability $\gamma = 0.5$. 
By varying $b$ in the range $[-1, 1]$, we illustrate in Fig.~\ref{fig:varyb} both the exact protection derived in Theorem~\ref{thm_protection_infinite_horizon} and the lower bound computed in Corollary~\ref{coro_protection_LB} with $\eta = 1$. 
The optimal choice for $b$ is sensitive to the problem parameters. 
For the first experiment we chose $\sigma = 1$ in \eqref{eq:regression} and we notice (c.f. Fig.~\ref{fig:sigma=1}) that the optimal value of $b$ is approximately $-0.55$ which results in the lowest protection. 
This plot also shows that, in general, negative $b$ values are preferred by the adversaries over positive ones.
To also demonstrate how sensitive the optimal value of $b$ can be, we only changed the parameter $\sigma$ to $0.1$, and the resulting plot (c.f. Fig.~\ref{fig:sigma=.1}) is significantly different where small positive $b$ values are preferred. 
This can be explained roughly as follows. Recall that the role of $b$ is to capture how $\xi_t$ decays overall. In the first case, the convergence time is much shorter with significant oscillations 
(in the components) of the state vector $x_t$ causing its time-difference $\xi_t$ to change signs frequently (c.f. Fig.~\ref{fig:varyb-state-trajectory}). Thus, to reflect this behavior, a value $b<0$ is needed. On the other hand, when $\sigma$ decreases in the second case, we practically reduce the Lipschitz constant of the objective function, leading to a much slower convergence. In this regime,  $x_t$ converges exponentially without oscillations most of the time. As a result, $\xi_t$ almost does not change signs and thus using $b>0$ would be more suitable.
These experiments also show that an unbiased estimate (i.e., $b=1$) might not be preferred over $b=0$.

\begin{figure*}
    \centering
    \includegraphics[scale = 0.38]{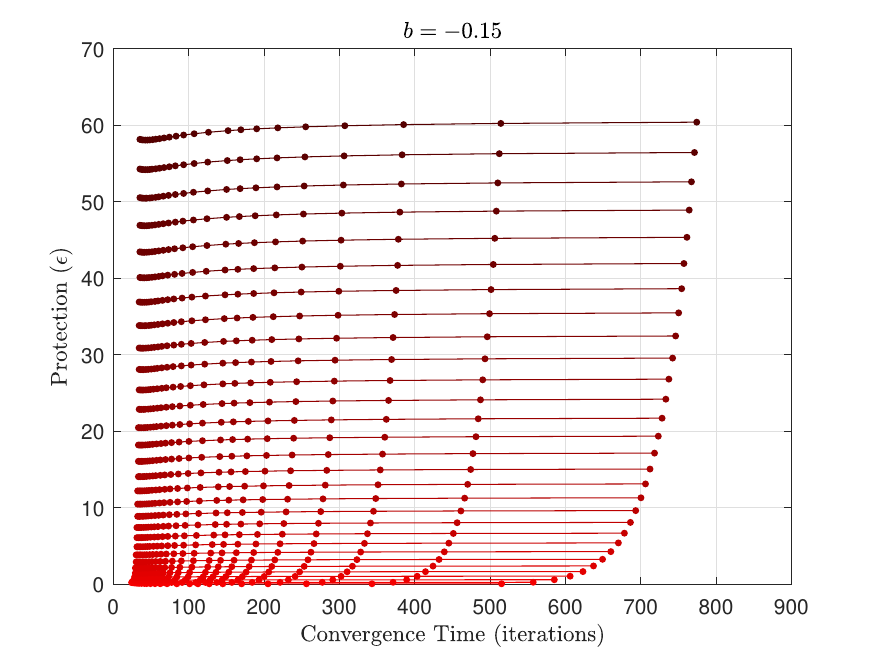}
    \includegraphics[scale = 0.38]{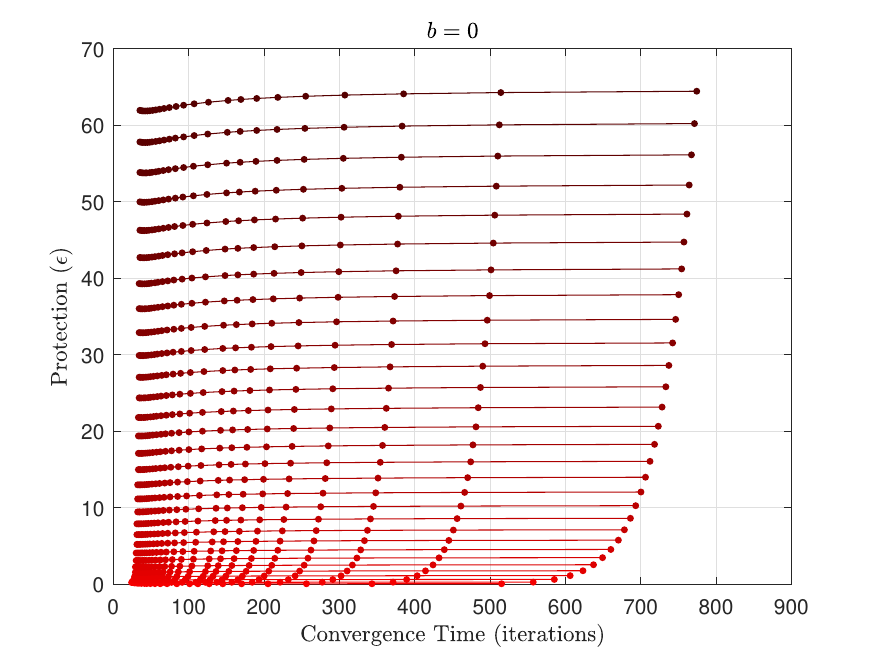}
    \includegraphics[scale = 0.38]{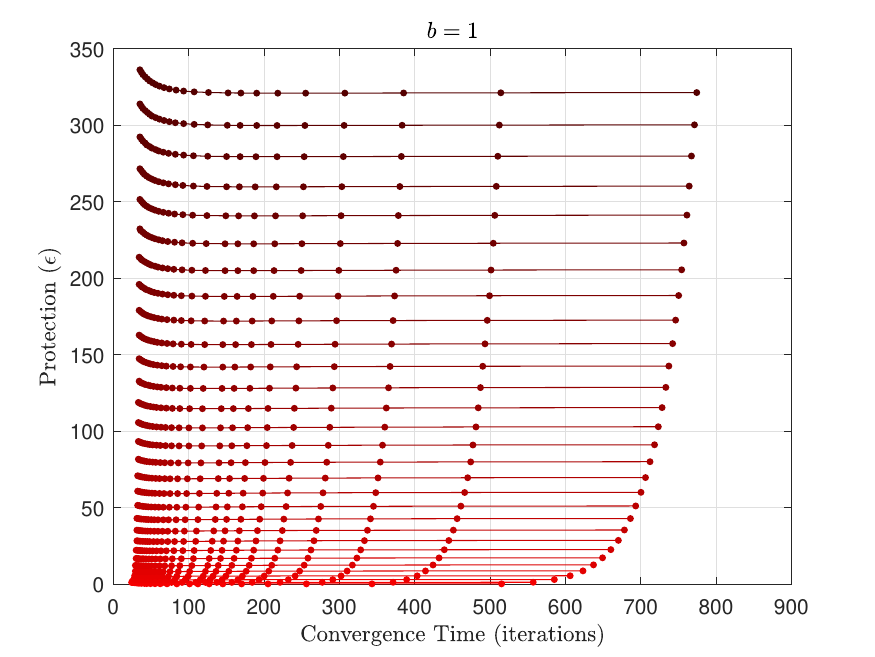}
    \put(-430,100){ $\xleftarrow{\alpha \text{   increases   }}$}
     \put(-360,30){\rotatebox{90}{$\xrightarrow{\|x_0\| \text{    increases   }}$}}
    % \put(-150,170){$\gets$ }
    \caption{Convergence speed vs. protection from Theorem \ref{thm_protection_infinite_horizon} by varying $\alpha$ and $x_0$.}
    \vspace{-3mm}
    \label{fig:protection and convergence speed}
\end{figure*}

Next, we validate the fact that $\|x_{i,0}\|$ as well as the algorithm parameters (e.g., $\alpha$) influence the amount of protection. 
The algorithm parameters control the trajectory taken by the node states $x_{i,t}$ and hence directly influencing the $\xi_{i,t}$ and the resulting protection. 
We fix $\sigma=1$ and vary the parameter $\alpha$ in \eqref{eq:extra} within the range $[0.01,0.2]$.
We choose one initial state vector $x_{\text{init}} \in \R^{nm}$ randomly and then scaled this initial state to generate $30$ initial state vectors $x_0 \in \{x_{\text{init}}, 1.5x_{\text{init}}, 2x_{\text{init}},\cdots,15.5x_{\text{init}}\}$.
For each pair of $(\alpha,x_0)$, we ran DICO and then computed two quantities: the protection amount $\epsilon$ and the convergence time defined as the number of iterations taken for the algorithm to converge to a value within $1\%$ of $x^*$.
Each line in Fig.~\ref{fig:protection and convergence speed} corresponds to a fixed $x_0$. 
As $x_0$ is scaled, the protection amount increases. 
The dots on a fixed color line represent different values of $\alpha$. 
As $\alpha$ increases, the dots move from right to left. 
 
From Fig.~\ref{fig:protection and convergence speed} we notice that, for a fixed value of $\alpha$, the amount of protection increases at the expense of convergence time when $\|x_0\|$ is increased. 
This shows a trade-off between convergence and protection for Algorithm~\ref{AL:ICC}. 
While the effect of $x_0$ on the protection (and convergence time) is somewhat straightforward from the expression in \eqref{eq_themorem_main}, that of $\alpha$, however, is not so obvious. 
Here, $\alpha$ impacts the trajectory of $\xi_i$, which is also dependent on the objective function $F$ in~\eqref{eqProblem}.
The effect of $\alpha$ on the convergence-protection curve is even more interesting. 
As $\alpha$ is increased for a fixed $\|x_0\|$,
the amount of protection slightly reduces, and then starts increasing with $\alpha$ when $b=0$ or $-0.15$. 
{This is because when $\alpha$ is sufficiently small, the algorithm converges exponentially without oscillations, 
i.e., $\|x_{k+1}-x_k\|^2 \approx C e^{-k \eta(\alpha)}, \forall k\ge 0$, where $C >0$ is a constant and $\eta(\alpha) >0$ is the convergence rate which increases with~$\alpha$. Thus, $Q \approx \|x_0\|^2+\frac{C}{1-e^{-\eta(\alpha)}}$, which decreases as $\alpha$ increases. Large values of $\alpha$ can lead to more oscillations, 
% (and even divergence), 
hence larger quadratic variations.}

\section{Conclusion}\label{secConclusion}

In this paper, we have studied a privacy issue of distributed systems against eavesdroppers that can intercept (successfully with some probability) communications with the goal of estimating the agent's final state. 
We show that the agents are unprotected in \textit{every} scenario where they are required to share their states. 
In contrast, by exchanging the innovation signals, the agent can harness the system-level protection that is inherently present in such systems. 
The proposed innovation-shared method is a complementary to existing approaches such as differential privacy.
One may use differential-privacy or encryption based methods along with our proposed method to obtain a higher amount of privacy than what would have been achievable from using only differential-privacy/encryption based methods.
Since our approach does not alter the accuracy of the converged solution, using it  in juxtaposition to other privacy preserving techniques will not incur further  accuracy loss.

Given the generic nature of our proposed method and the analysis, one may investigate particular problems (e.g., multi-agent consensus, rendezvous, distributed estimation)  and analyze the achievable protection for such problems under the innovation-shared communication scheme. 
In this work, we considered a class of distributed optimization problems as an example and demonstrate that the algorithm's parameters (e.g., $\alpha$, $x_0$) can in fact improve the achievable protection.
We show that there is a fundamental relation between the total quadratic variation of the innovation signal and the achievable protection.

\bibliographystyle{ieeetr}
\bibliography{ref}

% \newpage
\section{Appendix}

\begin{lemma}\label{lem_PQRS}
Suppose $Q:=\sum_{t\ge 0}\|\xi_t \|^2 < \infty$ and $|bc|<1$. Let $\rho \!=\! \frac{b\gamma}{1-c}$ and $\nu \!=\! b-\bar\gamma - \rho\gamma$. Then the  series $P \!=\! \sum_t\E e_t^T\xi_t, ~ 
R \!=\! \sum_t\E z_t^T\xi_t$ and $S \!=\! \sum_t\E\|z_t\|^2 $
% \begin{align*}
% P \!=\! \textstyle\sum_t\E e_t^T\xi_t, ~ 
% R \!=\! \textstyle\sum_t\E z_t^T\xi_t, ~ S \!=\! \sum_t\E\|z_t\|^2 
% \end{align*}
exist. In fact, $S \!=\! \textstyle \frac{\E\|z_0\|^2 + \gamma Q}{1-bc}$ and 
\begin{align} \label{eq_PR}
    \bar{b}P  \!+\! (b\!-\!\rho)R  \!=\! \inprod{\theta_0 }{x^* \!-\! x_0}  \!+\! \textstyle\frac{\nu}{2}(\|x^*\|^2 \!-\! \|x_0\|^2 \!-\! Q) 
\end{align} 
with $\theta_0 := \bar{b}\E e_0 + (b-\rho) \E z_0 - \nu x_0$.
\end{lemma}
\begin{proof} 
By \eqref{eq_TE_znorm} and conditions $|bc|\!<\!1$ and $Q\!<\!\infty$, we have
\begin{align*}
S 
&=\textstyle \sum_t (bc)^{t}\E\|z_0\|^2 + \gamma \sum_t \sum_{k=0}^t \|\xi_k\|^2 (bc)^{t-k} \nonumber\\
&=\textstyle \frac{\E\|z_0\|^2}{1-bc} + \gamma \sum_t \sum_{k}\|\xi_t\|^2 (bc)^{k} = \textstyle \frac{\E\|z_0\|^2 + \gamma Q}{1-bc}. %\label{eq_sumEzt2}
\end{align*} 

Next, we show that $R$ exists, i.e., $\sum_t \inprod{\E z_t}{\xi_t}$ is a convergent series. Since $2|\E z_t^T\xi_t| \le  \|\E z_t\|^2 + \|\xi_t\|^2 \le \E\| z_t\|^2  +  \|\xi_t\|^2$, 
it follows that $R$ is absolutely convergent and $2|R| \le S + Q$. 
Now consider $P$; let $P_t$, $Q_t$, and $R_t$ denote the partial sums. Then 
\begin{align}
    \bar{b}P_t  + bR_t = \textstyle\sum_{k=0}^t  \inprod{\bar{b} \E e_k + b \E z_k}{\xi_k}.\label{eq_PR_combination}
\end{align}
Let $u_t := \bar{b}\E e_t + b \E z_t.$  
Then, by \eqref{eq_TE_et} and \eqref{eq_TE_z} we have
\begin{align*}
    &u_{t+1} %= \bar{b}\E e_{t+1} + b\E z_{t+1} \\
    %= \bar{b}\E e_t + c \E z_t + (b+\gamma -1)\xi_t \\
    = u_t+ (c-b)\E z_t + (b - \bar{\gamma} )\xi_t.
\end{align*}
By multiplying both sides of \eqref{eq_TE_z} with $\rho = b\gamma/\bar{c}$ and then subtracting from the above relation, we obtain
\begin{align*}
    &u_{t+1} -\rho\E z_{t+1} = u_t - (\rho c+b\gamma) \E z_t + (b-\bar{\gamma} - \rho\gamma) \xi_t.
\end{align*}
Using $\rho = \rho c+b\gamma$ and $\nu = b-\bar{\gamma} - \rho\gamma$, we further have 
$u_{t} -\rho\E z_t 
= u_0 -  \rho \E z_0 + \nu\textstyle \sum_{k=1}^{t-1}\xi_k
= u_0 -  \rho \E z_0 + \nu(x_t-x_0).$
% \begin{align*}
%     u_{t} -\rho\E z_t 
%     &= u_0 -  \rho \E z_0 + \nu\textstyle \sum_{k=1}^{t-1}\xi_k\\
%     &= u_0 -  \rho \E z_0 + \nu(x_t-x_0).
% \end{align*}
Thus, $u_{t}\!-\!\rho\E z_t \!-\! \nu x_t \!=\! u_0 -  \rho \E z_0 -\nu x_0 \!=\!\theta_0$. 
As a result,
\begin{align*}
    &\bar{b}P_t  + bR_t = \textstyle\sum_{k=0}^t \inprod{u_k}{\xi_k} \tag{by \eqref{eq_PR_combination}}\\
    &=\! \textstyle \inprod{\theta_0}{\sum_{k=0}^t\xi_k} \!+\!\rho \textstyle\sum_{k=0}^t\! \inprod{\E z_k}{\xi_k} \!+\! \nu\sum_{k=0}^t\!\inprod{x_k}{\xi_k}
\end{align*}
Since $\xi_k = x_{k+1}-x_k$,  $\sum_{k=0}^t\xi_k = x_{t+1}-x_0$,
and 
\begin{align*}
    \textstyle 2\sum_{k=0}^t\inprod{x_k}{\xi_k} &= \textstyle\sum_{k=0}^t \big( \|x_{k+1}\|^2 - \|x_k\|^2 - \|\xi_k\|^2 \big) \\
    &= \|x_{t+1}\|^2 - \|x_0\|^2 - Q_t.
\end{align*}
Then, 
$\bar{b}P_t  + bR_t = \inprod{\theta_0 }{x_{t+1} - x_0} + \rho R_t  + \frac{\nu}{2}(\|x_{t+1}\|^2 - \|x_0\|^2 - Q_t).$ 
Letting $t\to\infty$ implies that $R$ converges and \eqref{eq_PR} holds.
\end{proof}

\subsection{Proof of \Cref{thm_protection_infinite_horizon} }\label{proof_main_theorem_infinite_horizon}
Using \eqref{eq_et1}, $\E\|e_{t+1}\|^2$ can be expanded as follows
\begin{align} \label{eq_TE_esquare}
    &\E\|e_t -\xi_t\|^2 +\E\|z_{t+1}\|^2  + 2\E [z_{t+1}\T(e_t - \xi_t)] \nonumber\\ 
    & = \E\|e_t -\xi_t\|^2 +\E\|z_{t+1}\|^2 \\
    &\quad +2\gamma \xi_t\T \E e_t -2\gamma\|\xi_t\|^2   -2c\xi_t\T \E z_t \nonumber\\
     &\quad + c \E\|z_t\|^2 + c \E\|e_t\|^2 - c\E\|e_{t-1} - \xi_{t-1}\|^2, \nonumber
\end{align}
where the last equality is obtained by noting that
\begin{align*}
    &\E [z_{t+1}\T(e_t - \xi_t)] = \E [ \E [z_{t+1}\T(e_t - \xi_t)~|~\mu_{t+1}] ] \\
    &= \gamma \E[\xi_t\T(e_t - \xi_t)] + c \E[z_t\T(e_t - \xi_t)] \\
    &=  \gamma \xi_t\T \E e_t -\gamma\|\xi_t\|^2  -c\xi_t\T \E z_t   \\
    & \qquad + \textstyle \frac{c}{2}\big[ \E\|z_t\|^2 + \E\|e_t\|^2 - \E\|z_t - e_t\|^2 \big] \\
    &\overset{\eqref{eq_et1}}{=}  \gamma \xi_t\T \E e_t -\gamma\|\xi_t\|^2  -c\xi_t\T \E z_t   \\
    & \qquad + \textstyle \frac{c}{2}\big[ \E\|z_t\|^2 + \E\|e_t\|^2 - \E\|e_{t-1} - \xi_{t-1}\|^2 \big].
\end{align*}

Now define $\delta_t = \E\|e_t\|^2 - \E\|e_{t-1} - \xi_{t-1} \|^2$ for $t\ge 1$ and $\delta_0 =\E \|e_0\|^2 - \|x_0\|^2$. 
Rearranging \eqref{eq_TE_esquare} yields
\begin{align*}
     \delta_{t+1} &= c \delta_t  +\E\|z_{t+1}\|^2 + 2\gamma \xi_t\T \E e_t \\\nonumber
    &\quad  -2\gamma\|\xi_t\|^2  -2c\xi_t\T \E z_t + c \E\|z_t\|^2. 
    %\label{eq_E_de}
\end{align*}
Using \eqref{eq_TE_znorm} to replace $\E\|z_{t+1}\|^2$ in the last equation yields
\begin{align} \label{eq_deltadynamics}
    \delta_{t+1} = c \delta_t + \phi_t = c^{t+1}\delta_0 + \textstyle\sum_{k=0}^t \phi_k c^{t-k}, 
\end{align}
with $\phi_t \!:=\! c(b\!+\!1)\!\E\!\|z_t\|^2 \!-\! \gamma\|\xi_t\|^2\! + \! 2\xi_t\T\! (\gamma\E \!e_t- c\E\!z_t)$. 
Now let $\Phi_t:=\sum_{k=0}^t \phi_k c^{t-k}$ 
and note that
\begin{align}
   &\E \|e_{t+1}\|^2  =  \delta_{t+1} + \E\|e_t-\xi_t\|^2 \nonumber\\
   &=  c^{t+1}\delta_0 + \Phi_t + \E\|e_t-\xi_t\|^2 \tag{by \eqref{eq_deltadynamics} }\\
   &=c^{t+1}\delta_0 + \Phi_t + \E\|e_t\|^2 -2\xi_t\T \E e_t +\|\xi_t\|^2. \nonumber
\end{align}
Unrolling this relation, we have 
\begin{align}
   \E \|e_{t+1}\|^2&= \textstyle \E \|e_0\|^2 + c\sum_{k=0}^t c^{k}\delta_0   + \sum_{k=0}^t\|\xi_k \|^2\nonumber\\
    &\quad \textstyle  + \sum_{k=0}^t\Phi_k - 2\sum_{k=0}^t\xi_k\T \E e_k \label{eq_etnorm_dynamic}
\end{align}
To find the limit, we will show that $\sum_{k}\Phi_k$ exists. In fact, by Lemma~\ref{lem_PQRS}, $P,Q,R$ and $S$ exist, and
\begin{align*}
    \textstyle \sum_t\phi_t &= c(b+1)S - \gamma Q + 2\gamma P - 2cR \\
    \textstyle \sum_t\Phi_t &= \textstyle \sum_{t}\sum_{k\le t} \phi_k c^{t-k} = \sum_{t,k}\phi_t c^{k} = \bar{ c}^{-1}\sum_t\phi_t.
\end{align*}
Thus, by \eqref{eq_etnorm_dynamic}, we have $\E \|e_{\infty}\|^2 = \E \|e_0\|^2 + \frac{c\E\delta_0}{\bar c} + A$
with 
\begin{align} 
    A &= \textstyle \sum_{t\ge 0} \big( \|\xi_t \|^2 + \Phi_t - 2\xi_t\T\E e_t\big)\nonumber\\
    &= Q -2P + \textstyle \frac{c(b+1)}{\bar c}S+ \frac{2\gamma P -\gamma Q - 2cR}{\bar c}.
\end{align}
Since $\delta_0 = \E\!\|e_0\|^2 - \|x_0\|^2$ and $\E\! \|e_0\|^2 = \bar\gamma \|z_{-1}-x_0\|^2$, %it then follows that 
\begin{align}\label{eq_AB}
    \textstyle \frac{\bar c}{\bar \gamma}\E \|e_{\infty}\|^2 = \frac{\bar c}{\bar \gamma}A + \|z_{-1}-x_0\|^2 - b \|x_0\|^2.
\end{align} 

We can expand $A$ as follows
\begin{align}
    \textstyle \frac{\bar c}{\bar \gamma}A &= \textstyle b(b+1)S + \bar b Q -  2\big( \bar bP  + bR\big) \label{eq_A_scaled}\\
    &= \textstyle \frac{b(b+1)}{1-bc}(\E\|z_0\|^2 + \gamma Q) + (\bar b+\nu) Q -2\rho R \nonumber\\
    &\quad \textstyle - \inprod{2\theta_0 + \nu x^*+ \nu x_0}{~x^*-x_0}, \nonumber
\end{align}
where the last term can be expressed as
\begin{align*}
    &\inprod{2\theta_0 + \nu x^*+ \nu x_0}{d_0} =\inprod{2\theta_0 + 2\nu x_0 + \nu d_0}{d_0}\\
    &= \textstyle 2\inprod{\theta_0 + \nu x_0}{d_0} +\nu\|d_0\|^2\\
    &= \textstyle 2\inprod{\bar b\E e_0 + (b-\rho) \E z_0}{d_0} +\nu\|d_0\|^2\\
    &= \textstyle \frac{2\bar b\bar\gamma}{\bar c}\inprod{\E z_{-1} -\bar b x_0}{d_0} +\nu\|d_0\|^2.
    %&= \textstyle 2(1-\frac{\gamma}{1-c})\inprod{\E z_{-1}-(1-b)x_0}{d_0} +\nu\|d_0\|^2
\end{align*}
Combining the relations above with \eqref{eq_AB} completes the proof.

\subsection{Proof of Corollary \ref{coro_protection_LB}}\label{proof_coro_protection_LB}
% We can  bound $R$ as follows: 
Note that 
$2\rho R \le  \textstyle\sum_t |\rho|\big( \eta^{-1}\|\E z_t\|^2 + \eta\|\xi_t\|^2 \big) \le |\rho|\eta Q + \sum_t |\rho|\eta^{-1}\E\| z_t\|^2 = |\rho|(\eta^{-1}S + \eta Q)$ for any $\eta>0$, where the first equality follows from Cauchy-Schwartz inequality and second one from $\|\E z_t\|^2\le \E\| z_t\|^2$. It remains to use \eqref{eq_AB}--\eqref{eq_A_scaled} and note that these bounds are tight when $\rho=0$.

\subsection{Proof of Corollary \ref{coro_protection_unbiased}}\label{proof_coro_protection_unbiased}
Let us compute $\E e_\infty$. From  \eqref{eq_et1}, we have
$\E e_{t+1} = (\sum_{k=-1}^t \E z_{k+1}) - x_{t+1}.$ 
Now, using \eqref{eq_TE_z}, we also have $\sum_{k=-1}^t \E z_{k+1}  = \sum_{k=-1}^t (c\E z_k + \gamma \xi_k)$, which implies
\begin{align*}
\textstyle\sum_{k=-1}^t \E z_{k+1} & = \textstyle\frac{c}{1-c} (\E z_{-1} - \E z_{t+1}) + \frac{\gamma}{1-c}x_{t+1}
\end{align*}
Thus, $\E e_{t+1} =  \textstyle\frac{c}{\bar c} (\E z_{-1} - \E z_{t+1}) - \frac{\bar{b}\bar{\gamma}}{\bar c}x_{t+1}$. Since $x_t\to x^*$ and $\E z_t \to 0$ as $t\to \infty$, we have
$\E e_\infty = 
\textstyle\frac{\bar\gamma}{\bar c} (b\E z_{-1} - \bar{b}x^*)$. 
Given $\E x_{-1}= 0$ and $b = 1$, we have $\E e_\infty = 0$. 

We now find $\E \|e_{\infty}\|^2$. It follows from \eqref{eq:b=1_case} and Lemma~\ref{lem_PQRS} that 
$\frac{\bar c}{\bar \gamma}\E \|e_{\infty}\|^2 = 2\|x_0\|^2 + 2Q - 2R \overset{(i)}{=} 2\|x_0\|^2 + Q - S + \sum_t \|\E z_t - \xi_t\|^2 \overset{(ii)}{=} \|x_0\|^2 + \sum_t \|\E z_t - \xi_t\|^2$, 
where $(i)$ is obtained by using the definitions of $Q,R$ and $S$ from Lemma~\ref{lem_PQRS} and observing that $Q-2R + S = \sum_{t\ge 0} \|\E z_t - \xi_t\|^2$ when $z_{-1} = 0$. 
Finally, $(ii)$ is obtained by using the relationship that $S = \textstyle \frac{\E\|z_0\|^2 + \gamma Q}{1-bc}$ from Lemma~\ref{lem_PQRS} along with the fact that  $\E\|z_0\|^2 = \gamma \|x_0\|^2$ when $z_{-1} = 0$.

\end{document}